\documentclass[UKenglish, cleveref, autoref]{llncs}
\usepackage{algorithm}
\usepackage{algorithmic}
\usepackage{color}
\usepackage{graphicx}
\usepackage{amsmath, amssymb}
\usepackage{ulem}
\usepackage{soul}
\usepackage{hyperref}

\newcommand{\G}{\ensuremath{{\mathcal G}}}
\newcommand{\T}{\ensuremath{{\mathcal T}}}
\newcommand{\D}{\ensuremath{{\mathcal D}}}
\newcommand{\thick}{{\sf thick}}

\newcommand{\embed}[1]{{\sf embed}(#1)}
\newcommand{\crop}{{\sf crop}}

\newcommand{\territory}{{\sf territory}}
\newcommand{\parent}{{\sf parent}}
\newcommand{\child}{{\sf child}}
\newcommand{\orient}{{\sf orient}}
\newcommand{\diam}{{\sf diameter}}

\renewcommand{\le}{\leqslant}
\renewcommand{\ge}{\geqslant}
\renewcommand{\geq}{\geqslant}
\renewcommand{\leq}{\leqslant}

\newcommand{\para}[1]{\vspace{0.2em}\noindent\textbf{#1.}~}

\newenvironment{proofsketch}{%
  \proof}{\endproof}

\newboolean{short}
\setboolean{short}{false}

\newcommand{\shortOnly}[1]{\ifthenelse{\boolean{short}}{#1}{}}
\newcommand{\onlyShort}[1]{\ifthenelse{\boolean{short}}{#1}{}}
\newcommand{\longOnly}[1]{\ifthenelse{\boolean{short}}{}{#1}}
\newcommand{\onlyLong}[1]{\ifthenelse{\boolean{short}}{}{#1}}

\begin{document}
\title{Guarding A Polygon Without Losing Touch}
\author{
Barath Ashok\inst{1} \and
John Augustine\inst{1} \and
Aditya Mehekare\inst{2} \and
Sridhar Ragupathi\inst{2} \and
Srikkanth Ramachandran\inst{2} \and
Suman Sourav\inst{3}
}
\institute{
Indian Institute of Technology, Madras \and
National Institute of Technology, Tiruchirappalli \and
Advanced Digital Sciences Center, Singapore
}

\maketitle
\begin{abstract}
We study the classical {\it Art Gallery Problem}  first  proposed by Klee in 1973 from a mobile multi-agents perspective. Specifically, we require an optimally small number of agents (also called guards) to navigate and position themselves in the interior of an unknown simple polygon with $n$ vertices such that the collective view of all the agents covers the polygon. 

We consider 
the {\it visibly connected setting} wherein agents must remain connected through line of sight links -- a requirement particularly relevant to multi-agent systems.  
We first provide a centralized algorithm for the visibly connected setting that runs in time $O(n)$, which is of course optimal. We then provide algorithms for two different distributed settings. 
In the first setting, agents can only perceive relative proximity (i.e., can tell which of a pair of objects is closer) whereas they can perceive exact distances in the second setting.
Our distributed algorithms work despite agents having no prior knowledge of the polygon. 
Furthermore, 
we provide lower bounds to show that our distributed algorithms are near optimal. 

Our visibly connected guarding ensures that (i)  the guards form a connected network and (ii) the polygon is fully guarded. Consequently, this guarding provides  the distributed infrastructure to execute any geometric algorithm for this polygon.


\end{abstract}

\section{Introduction}

The {\it Art Gallery Problem} is a classical computational geometry problem that seeks to minimize the number of guards (or agents in our context) required to guard an art gallery (represented by a simple polygon $P$ comprising  vertices $\{p_1, p_2, \ldots, p_n\}$). To successfully guard the gallery, every point inside the polygon must be visible to at least one guard, i.e., for every point in the polygon, there must exist at least one guard such that the segment joining the point to the guard does not  intersect the exterior of the polygon. 
This computational geometry problem was first posed by Klee in 1973, and thereafter has been widely studied over the years (see \cite{O'Rourke:1987:AGT:40599,163407,URRUTIA2000973,G07}). 

In this paper, we investigate  a variation of the problem called the {\it visibly connected art gallery problem} from a distributed multi-agents perspective. We require an optimally small number of agents (also called guards) with omni-directional vision to navigate an unknown simple polygon with $n$ vertices in a coordinated manner and position themselves in its interior such that the collective view of all the agents covers the polygon. Additionally, for the connected art-gallery problem (as in \cite{OGB11}), it is required that the agents maintain  line-of-sight connectivity. More precisely, the visibility graph~\cite{Liaw1993TheMC,OGB11} comprising agents as nodes and edges between agents that are within line of sight of each other (unobstructed by polygon edges) must be a connected graph.  

The visibly connected art gallery problem was studied as early as 1993 by Liaw {\it et al.}~\cite{Liaw1993TheMC} in the centralized setting, but only for the special case of spiral polygons. 
Hernandez-Penalver~\cite{hernandez1994controlling} considered simple polygons and showed that $\lfloor n/2 \rfloor-1$ guards are sufficient and sometimes necessary. Pinciu~\cite{P03} presented a  centralized algorithm based on iteratively processing the dual graph of the polygon's triangulation. Although~\cite{P03} lacks the analysis, one can infer that it runs in time linear in $n$, the number of vertices of the polygon. However, the algorithm is somewhat complicated and not amenable for parallel or distributed computing. Obermeyer et al. \cite{OGB11} provided a distributed algorithm that is capable of handling polygons with holes, but unfortunately requires $O(n^2)$ rounds. 

Our work is motivated by recent advancements in unmanned aerial vehicles (UAVs), especially those capable of automated sensing (either through photogrammetry or LiDAR) and communication (typically through line-of-sight electromagnetic radio waves). Such UAVs are typically deployed into unknown territories from which they are required to navigate, learn, and perform useful tasks in a coordinated manner. We model these UAVs as point agents that start  from a common starting point assumed without loss of generality (w.l.o.g.) to be $p_1$. Agents operate in synchronous rounds during which they can look, compute, communicate and move. They are required to coordinate with each other and achieve full visibility coverage of the polygon while maintaining line-of-sight connectivity with each other.
 
A connected visibility graph ensures that there exists a path between every pair of guards. So visibly connected guards can simultaneously maintain  coverage of the polygon and execute distributed computing protocols through line of sight communication. 

\enlargethispage{\baselineskip}
\para{Our Contributions}
We begin with a description of a centralized algorithm in Section~\ref{sec:central} that takes a simple polygon $P$ with $n$ vertices as input and produces a placement of at most $\lfloor n/2 \rfloor-1$ guards that satisfy the requirements of the visibly connected art gallery problem. This algorithm only requires $O(n)$ time. 
Here, we introduce a notion of  triplets (three connected nodes) in the weak dual graph $\D$ (defined formally in Section~\ref{sec:prelims}). Informally, $\D$ is the graph whose nodes are triangles of a triangulation of $P$ and arcs connect pairs of triangles that share an edge. We show that $\D$ can be decomposed into $O(n)$  triplets that are connected and cover $\D$, and then we compute a set of visibly connected guards by placing guards -- one per triplet -- positioned strategically within the triangles pertaining to each triplet.  Although Pinciu~\cite{P03} has already presented an $O(n)$ algorithm, we believe that our algorithm is simpler and,  more importantly, amenable to parallel computing. In particular, we show that our algorithm can be adapted to run in the PRAM model in time linear in the diameter of $\D$.

Subsequently, we turn our attention to distributed computing models in which the agents are independent mobile computing entities that must interact with each other to solve the visibly connected art gallery  problem. We define two model variants based on agents' perception capabilities. The {\it depth perception} variant wherein the agents can accurately perceive depth (i.e., distances) is inspired by UAVs with LiDAR technology~\cite{mcmanamon2019lidar}. On the other hand, the {\it proximity perception} variant only provides the agents with relative proximity. For concreteness in the model, we limit the proximity perception variant to being able to sense which of any two objects (i.e., edges or vertices of the polygon) is closer. It is inspired by photogrammetry~\cite{ABER2019}, which is cheaper and only guarantees coarser perception.  

We present algorithms for both cases. The algorithm for the proximity perception variant, presented in Section~\ref{sec:algo1}, operates by exploring visible territories within the polygon (formally defined in Section~\ref{sec:prelims}). We describe two forms of exploration, one in a breadth-first manner and the other in a depth-first manner. Since the polygon structure is completely unknown to the agents, for each level of the breadth-first exploration, the nodes must communicate to the root in order to ensure that a sufficient number of agents are provisioned to explore that level. Our algorithm -- taking the best out of both explorations --runs in $O(\min(\tilde{d}_v^2, n))$ rounds. The term $\tilde{d}_v$ is a natural notion of diameter of $P$ called {\it minimal v-diameter} that we define formally in Section~\ref{sec:algo1}.  Informally, it is the largest diameter among all visibility graphs pertaining to minimal placement of connected guards that cover all vertices in $P$. The candidate placements are minimal in the sense that the removal of guards either leads to lack of coverage or loss of connectivity.

When the agents can perceive depth, we exploit this capability to place agents based on the medial axis of the polygon $P$ (defined formally in Section~\ref{sec:prelims}), which is a well-known tree-like structure that captures the ``shape'' of the polygon. Since depth perception is more powerful than proximity perception, we can take the best of all options to ensure a running time of $O(\min(\tilde{d}_v^2, D^2, n))$ time, where $D$ is the (unweighted) diameter of the medial axis tree.
Our algorithms require at most $n$ agents for an initial placement, but a subsequent post-processing ensures that at most $\lfloor n/2 \rfloor-1$ agents are placed in the polygon, which is optimal. %

To complement our complexity claims, we consider the weaker problem wherein the robots are not required to be placed in a visibly connected guarding position, but rather just that the entire polygon must be explored by the agents. The exploration problem only requires that for every point in $P$, some agent must have been within line of sight of that point at some time instant during the course of the algorithm. Clearly, any solution to the visibly connected guard placement problem will also be a solution for the exploration problem. Thus, we focus on showing a lower bound for the exploration problem. Specifically, we show that, for any deterministic  algorithm to even centrally coordinate $n/4$ agents to explore an initially unknown polygon, we can construct a polygon %
that requires $\Omega(D^2)$ (or $\Omega(\tilde{d}_v^2)$) communication rounds even with depth perception.

\onlyShort{Unfortunately, due to space limitation, we have deferred some of our proofs to the long version, which is available in full in the Appendix.}

\para{Related Work} 
The classical art gallery problem was first introduced by Klee in 1973. Chv\'{a}tal \cite{CHVATAL197539} showed by an induction argument that $ \lfloor\frac{n}{3}\rfloor$ guards are always sufficient and occasionally necessary for any simple polygon with $n$ vertices. Fisk \cite{FISK1978374} proved the same result via an elegant coloring argument. Lee and Lin in \cite{1057165} proved that determining a set of minimum number of guards that can guard a given polygon is NP-hard. Consequently, researchers have focussed on approximate solutions starting from an $O(\log n)$ approximation provided by Ghosh~\cite{G87} in 1987, along with a conjecture that the problem admitted a polynomial time constant approximation algorithm. However, Eidenbenz et al.~\cite{Eidenbenz2001} showed that the problem was APX-hard, thereby precluding the possibility of a PTAS unless P=NP. After several improvements over the years, in 2017, Bhattacharya et al.~\cite{BGP17} have reported constant factor approximation algorithms for the classical art gallery problem as well as for several well-studied variants. 

In literature, based on the different restrictions placed on the shape of the galleries or the powers of the guards, several variations of ``art gallery problems" have been studied. See \cite{O'Rourke:1987:AGT:40599}, \cite{163407}, and \cite{URRUTIA2000973} for details.

The {\it connected art-gallery problem} %
was first introduced by Liaw et al. in 1993 \cite{Liaw1993TheMC}, where they refer to the problem as \textit{minimum cooperative guards problem} and study it on $k$-spiral polygons (polygons with a maximal chain of $k$ consecutive reflex vertices, i.e., vertices with internal angle $> 180$ degrees). %
It was also shown \cite{Liaw1993TheMC} that this problem is NP-Hard for  simple polygons but can be solved in linear time in spiral and $2$-spiral polygons (also see \cite{sarasamma} for results on $k$-spiral graphs). For simple polygons with $n$ vertices,  Hern\'{a}ndez-Penalver  \cite{hernandez1994controlling} proved by induction that $\lfloor\frac{n}{2}\rfloor -1$ guards are always sufficient to obtain a connected guarding. Moreover, they also show that $\lfloor\frac{n}{2}\rfloor -1$ guards are necessary for some polygons. The same result was also shown by Pinciu via an elegant coloring argument in  \cite{10.1007/3-540-45066-1_20}. 

In \cite{Liaw:1994:OAS:194679.194685}, Liaw et al. relax the strong connectivity condition from \cite{Liaw1993TheMC} and consider the case where %
there are no isolated vertices in the guards visibility graph. This problem of guarded guards where the overall connectivity of the guards visibility graph is non-essential has also been studied in \cite{Michael:2003:AGT:945968.945973,P03}.

In the distributed setting, Obermeyer et al. \cite{OGB11} study this problem in polygonal environment with holes.  They first design a centralized incremental partition algorithm (defined therein) and from that obtain the distributed deployment algorithm by a distributed emulation of the centralized algorithm. The authors give a deployment of agents that is guaranteed to achieve full visibility coverage of the polygon with $n$ vertices and $h$ holes in $O(n^2+nh)$ time, given that there are at least $\lfloor \frac{n+2h-1}{2}\rfloor$ agents. %
This work closely relates to our work. While the scope of \cite{OGB11} includes polygons with holes, their algorithm is not optimized for time and their ideas lead to algorithms that require $O(n^2)$ communication rounds even for simple polygons without holes. Notable prior works with ideas leading to~\cite{OGB11} can be found in \cite{10.1007/978-3-642-22935-0_18,1656416,4283034,Ganguli:2007:MCM:1415363}.

\onlyLong{Another interesting related work is Fekete et al. \cite{10.1007/978-3-642-22935-0_18} wherein they study the same problem with the added constraint that the guards have limited visibility. Specifically guards are visible to each other only if their distance is less than $r$. They study two problems, the first, Minimum Relay Triangulation Problem (MRTP), where they require minimum number of guards to be placed and the second, Maximum Area Triangulation Problem (MATP), where the number of guards is restricted to $n$ (the number of vertices of the polygon) and the objective is to cover maximum area of the polygon. They show that MRTP is NP-Hard and provide a PTAS for MATP. While their algorithm also maintains visible connectivity, more than $n$ guards maybe required because of the distance constraint. They show that offline MRTP and MATP problems are NP-hard and provide a PTAS for both. They also provide lower bounds for the competitive ratio of online MRTP problem ($6/5$) and a strategy that achieves $3$. For the online MATP problem they show that no constant lower bound exists for competitive ratio.}

\para{Organization of the Paper} In Section~\ref{sec:prelims}, we present some preliminary definitions including several geometric definitions pertaining to polygons as well as formal definitions of the distributed computing models. In Section~\ref{sec:central}, we provide a centralized algorithm to solve the visibly connected art gallery problem and then show how to parallelize it. In Section~\ref{sec:algo1} and Section~\ref{sec:algo2}, we present distributed algorithms under proximity perception and depth perception, respectively. We complement our upper bounds with a lower bound that is proved in Section~\ref{sec:lb}. We then conclude with some remarks and future works in Section~\ref{conclusion}.

\section{Preliminaries} \label{sec:prelims}
 Let $P$ be a simple polygon with $n \ge 4$ vertices; we use $P$ to refer to the polygonal region including both the interior and the boundary. In general, for a polygonal region $K \subseteq P$, we use $\partial K$ for the set of vertices of $P$ that lie on the boundary of this polygonal region. The ordered list of vertices of $P$ are denoted $p_1, p_2, \ldots, p_n$, and thus, $\partial P \triangleq \{p_1, p_2, \ldots, p_n\}$. Each open line segment connecting $p_i$ to $p_{i({\sf mod}\ n) + 1}$, $ 1 \le i \le n$, is denoted $e_i$. 
 We assume that the vertices are in general position, i.e., (i) no three vertices  are collinear, and (ii) no four vertices are co-circular.  We use the term {\it object} to refer to either a vertex or an edge. Thus, the objects of $P$ are $\{p_i\}_i \cup \{e_i\}_i$.

We use the notation $g \in P$ for some point $g$ to indicate that $g$ can either be a vertex, lie on an edge, or lie in the interior of $P$. Two points $g_1 \in P$ and $g_2 \in P$ are said to be in 
 line of sight of each other if the open line segment $\overline{g_1g_2}$ lies entirely within $P$. 
 We use $V_g^P$ (or just $V_g$ when $P$ is clear from context) to denote the visibility polygon of a point $g \in  P$, which is defined  as the subset of $P$ that contains all points that are in line of sight from $g$. See Figure~\ref{fig:visibility} for an illustration. 

We also borrow a useful definition from Obermeyer {\it et al.}~\cite{OGB11} for {\it vertex-limited visibility polygon} $\bar{V}_g^P$ for a point $g \in P$ w.r.t. $P$, which is a modified form of $V_g^P$. Notice that $V_g^P$ could have vertices that are not vertices in $P$; call such vertices {\it spurious vertices}. To get $\bar{V}_g^P$, perform the following operation repeatedly until there are no more spurious vertices: pick a spurious vertex $v$ with predecessor $p$ and successor $s$ and crop the visibility polygon by cutting along the line segment $\overline{ps}$ and removing the portion that lacks the point $g$ from further consideration. Note that either the predecessor or the successor may themselves be spurious. In Figure~\ref{fig:visibility}, note the first vertex $v$ that is clipped has successor $s$ that is itself a spurious vertex. An edge that is in $\bar{V}_g^P$ but not in $V_g^P$ is called a {\it gap edge}.
\begin{figure}[h]
	\centering
		\includegraphics[width=0.70\textwidth,page=18,clip=true,trim=0 140 25 100]{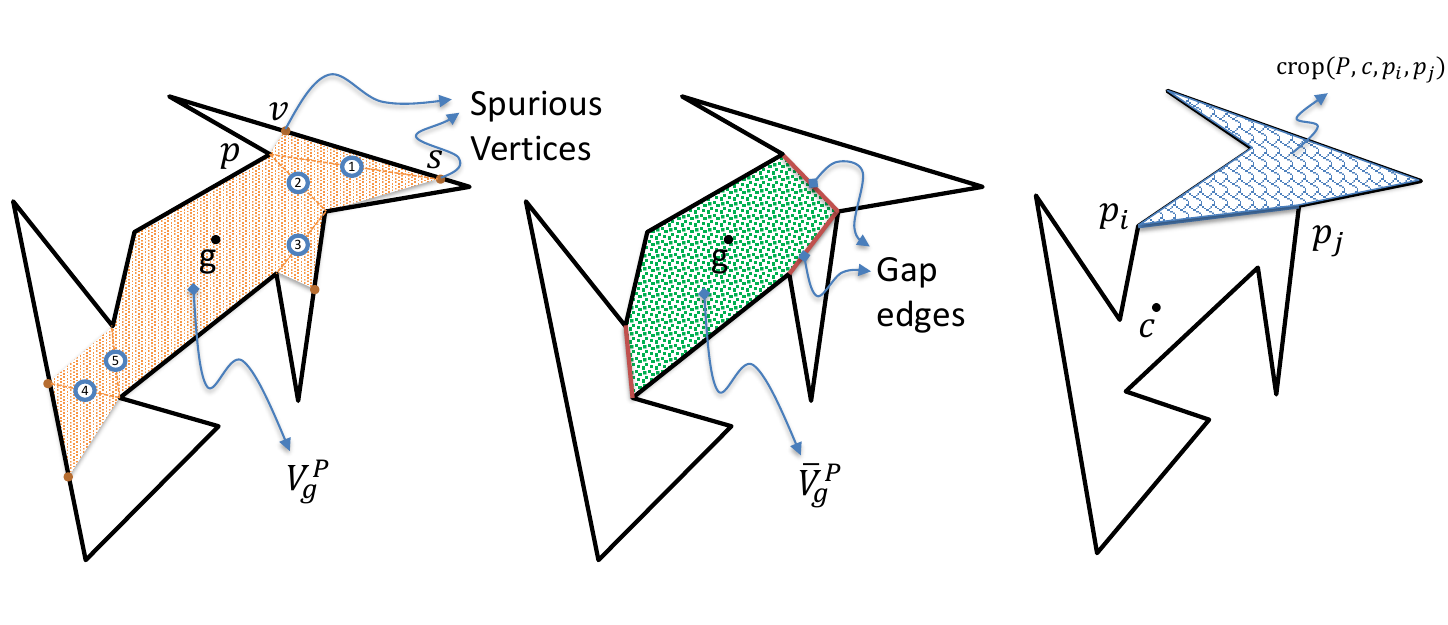}
	\caption{\small{Visibility polygon (left) and vertex-limited visibility polygon (at the center). The numbered line segments refer to one possible repeated sequence of cuts to arrive at vertex-limited visibility polygon. The polygon on the right depicts a cropped polygon.}}
	\label{fig:visibility}
\end{figure}
We also define a way to crop a polygon (cf. Figure~\ref{fig:visibility}). Formally, for any  pair of vertices $p_i$ and $p_j$ such that $\overline{p_i p_j}$ lies entirely within the interior of $P$ and $c \in P \setminus \overline{p_i p_j}$, we define  $\crop(P, c, p_i, p_j)$ to be the subset of $P$ obtained by cutting along $\overline{p_i p_j}$ and discarding the part that contains $c$.

\begin{definition}
Let $G = \{g_1, g_2, \ldots\}$ be a set of points in $P$. We say that the points in $G$ {\it guard polygon} $P$ if $\cup_{g \in G} V_g = P$.  In this context, we call the points in $G$ as guards of $P$.
\end{definition}

The classical art gallery problem seeks to find a smallest possible set $G$ of points that guard $P$, with variants including vertex and edge guarding.  
In this paper, we  consider a variant called the {\it connected art gallery problem} that, to the best of our knowledge, was introduced first by Liaw et al.~\cite{Liaw1993TheMC}. In this variant, guards are connected in a suitable way, which we now formalize. We define the {\it visibility graph} of a set of points $G$ within (and with respect to)  $P$, denoted $\G_G^P$ (or just $\G$ when clear from context) as the graph with vertex set $G$. Two points in $G$ are connected by an edge in $\G$ iff they are visible to each other within $P$.  A set $G$ of points in $P$ is said to be connected (w.r.t. $P$)  if $\G_G^P$ is connected.  
In the connected art gallery problem, we are required to compute a set $G$ of points that  guards $P$ {\it and} the additional requirement that $\G_G^P$ is connected. 
It is well-known~\cite{bcko08} that at most $\lfloor n/3 \rfloor$ guards are always sufficient to guard any polygon with $n$ vertices. However, this bound does not hold under  connected guarding.
\begin{claim}[consolidated from \cite{10.1007/3-540-45066-1_20} and \cite{P03}] \label{clm:lb}
There exist orthogonal polygons with $n$ vertices that require at least $ n/2 - 2$ number of connected guards even if we only require them to guard the vertices of the polygon~\cite{P03}. This bound  increases mildly to $\lfloor n/2 \rfloor -1$  
for simple polygons with non-orthogonal edges \cite{10.1007/3-540-45066-1_20}. 
\end{claim}

We now define several structures associated with any polygon $P$. A line segment joining two vertices is said to be a {\it diagonal} if its interior lies entirely within the interior of $P$. It is easy to see that a maximal set of diagonals that do not intersect each other decomposes the polygons into a set of $n-2$ triangles called a {\it triangulation} of $P$~\cite{bcko08}. A famous result by Chazelle~\cite{C91} shows us how to find such a triangulation in $O(n)$ time. Given a triangulation $T$ for a polygon $P$, the {\it weak dual graph} $\D_T^P$ (or just $\D$ when clear from context) is the graph (or more informatively, a tree) whose nodes are the triangles in $T$ with edges in $\D_T^P$ between pairs of triangles that share a common triangle edge. Note that the weak dual graph is a tree where each tree node has a degree of at most $3$. \onlyShort{An explanation as to why it is true has been postponed to the appendix.}

\onlyLong{For simple polygons, we can build a triangulation by dividing or splitting the polygon across some internal diagonal into two {\it disjoint} pieces (which in turn are simple polygons), recursively computing the triangulation for each split piece until the split pieces are also triangles. The input polygon's triangulation can be computed by gluing the two sub-solutions across this splitting diagonal. Notice that, by an appropriate sequence of choice of splitting diagonals at each step, any particular triangulation can be obtained as a solution. Moreover, this recursive procedure can be modified to obtain the {\it weak dual graph} as well, which, when combined with the fact that simple polygons always yield disjoint pieces when split across an internal diagonal gives us almost immediately that the {\it weak dual graph} is a tree. Additionally, since at most three triangles can share edges with a particular triangle in any given triangulation, it implies that the degree of any node of the {\it weak dual tree} is bounded by three.}

Recall that, the term {\it object} refers to either a vertex or an edge of the polygon $P$.
We define the medial axis $M$ of the polygon $P$ to be the (infinite) collection of points within $P$ that are equidistant from at least two distinct objects of $P$\onlyLong{(cf. Figure~\ref{fig:medial})}. 

\begin{claim}
For any simple polygon $P$, the  medial axis is a tree whose leaves are convex vertices of $P$, i.e., vertices with internal angle being less than 180 degrees.
\end{claim}
\onlyLong{
\begin{proofsketch}
The proof follows in a similar vein in relation to the {\it weak dual graph is a tree} claim - except here, the decomposition of the domain into sub-domains and the consequent gluing together of the medial axis of these sub-domains to give the overall solution requires a considered application of the domain decomposition lemma (see section 5 of \cite{choi1997mathematical}). Moreover, no concave vertex $v$ can be part of the medial axis, since any disc centered at $v$ must necessarily intersect with the outside of the polygon which means that there can be no maximal disc that lies entirely within the polygon centered at $v$. And further, using this observation, we can argue that $v$ cannot be the limit point or closure of any branch of the medial axis as well. Additionally, each convex vertex would have an angle bisector starting from it, which would be a part of the medial axis, and thereby taking the closure of the points on the medial axis makes these convex vertices the leaves of the medial axis.
\end{proofsketch}
}
Note that reflex vertices (i.e., vertices with internal angles greater than 180 degrees) cannot be nodes in the medial axis tree. An edge in $M$ is a non-empty maximal set of points that are equidistant between the same set of objects. When the two objects are of the same type (either both polygonal edges or both vertices), the corresponding medial axis edge will be a straight line segment. On the other hand, if one of the objects is a vertex and the other is a polygonal edge, the corresponding medial axis edge will be a parabolic arc. The endpoints of medial axis edges are the {\it medial axis nodes} or just nodes. Since the number of leaves is at most $n$, we get: 
\begin{claim}
The number of nodes in the medial axis of $P$ will be $O(n)$. 
\end{claim}
We use $D^P$ (or just $D$ when clear from context) to refer to the unweighted diameter of the medial axis $M$. More precisely, $D^P$ is the maximum number of edges over all paths in the medial axis tree of $P$.
\onlyLong{
\begin{figure}[h]
	\centering
		\includegraphics[width=0.7\textwidth,page=13,clip=true,trim=40 115 20 140]{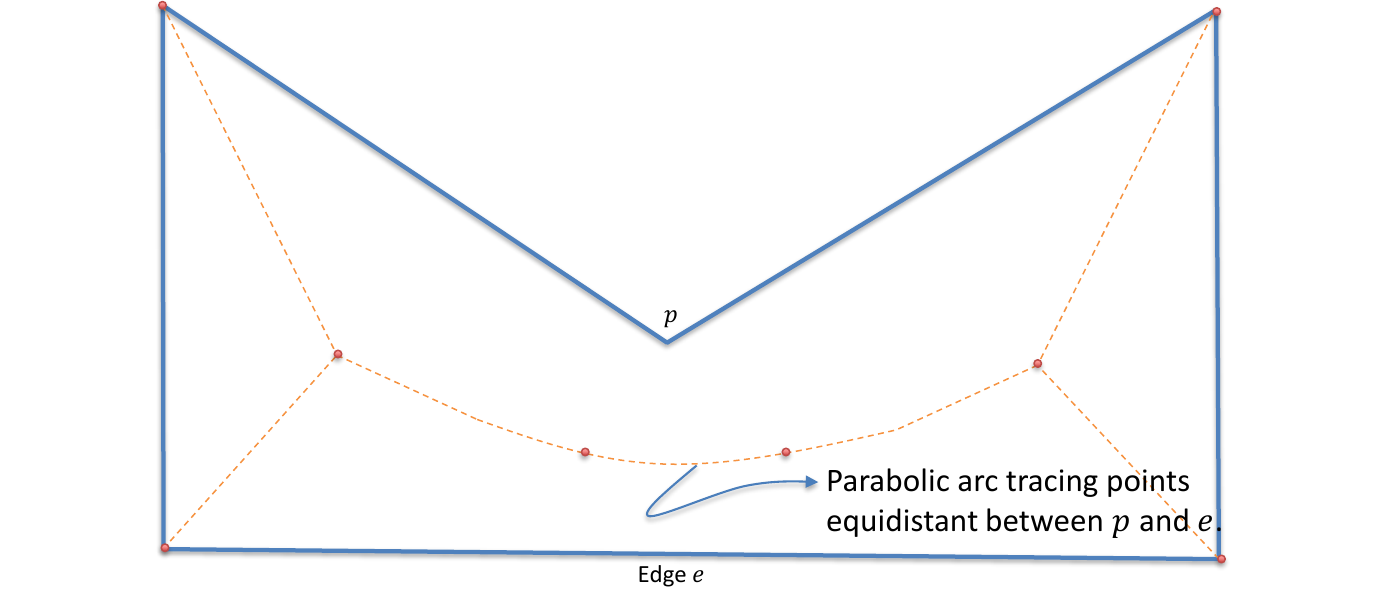}
	\caption{Medial axis of a polygon.}
	\label{fig:medial}
\end{figure}
}

\subsection{Computational Models.}

In this section, we focus on the connected art gallery problem in the classical sequential setting. In this case, we assume that the sequence of vertex points $(p_1, p_2, \ldots, p_n)$ are given in order, say, as an array of points. 

However, the connected art gallery problem is inspired by mobile agents that operate in a spatially distributed setting. So we employ a  distributed computing model based on the model used by Obermeyer {\it et al.}~\cite{OGB11}. For clarity, we assume a synchronous model with time discretized into a sequence of rounds and the agents execute a look-communicate-move cycle in each synchronous round; local computation is allowed at any time interspersed between the look-communicate-move cycles. %

We assume for the sake of convenience that there are $n$ agents (although our algorithm can operate with fewer). Agents (modeling mobile robots)  can be represented as points in the plane and as a result multiple agents can be co-located at the same point. Without loss of generality, assume all agents start at the same vertex somewhere in $P$.  Furthermore, agents can only move from one vertex $p_i$ to another vertex $p_j$ provided $\overline{p_i p_j}$ is a diagonal in $P$ (i.e., $p_i$ and $p_j$ have direct line of sight to each other). We assume that agents have unique IDs from $\{1, 2, \ldots, n\}$.  Each agent $g$ performs the following tasks within each round.

\noindent {\bf Look.} The agent $g$ first orients itself to start from a particular direction  (in the direction of another vertex called its {\it orientation vertex}) and perform a 360 degrees clockwise  sweep during which it creates a view of $\bar{V}_g^P$. The level of information  that the agent can gather depends on whether the agents have depth perception or not. With depth perception, the view is simply the full visibility polygon $\bar{V}_g^P$.
 Without depth perception, however, the view is limited to a sequence of alternating vertices and edges (possibly gap edges)  starting from its orientation vertex. In both cases, $g$ can also see other agents that are inside $V_g$.

\noindent {\bf Communicate.} Two agents can communicate as long as they are visible to each other (which of course includes co-located agents). Communication is via message passing. Each agent can send at most one message to each  agent that it can see.

\noindent {\bf Move.} This step again differs based on whether agents can perceive depth or not. Let us first consider the case when agents can perceive depth. Based on the outcome of the communication and computation, each agent $g$ chooses to move from its current location to a new location within its current visibility polygon. We assume that the agent -- once it reaches its destination location -- can ``remember'' its source position in the sense that it can spot the source location in its  view after it reaches its destination. The only restriction when agents cannot perceive depth is that they are limited to moving to vertices of the polygon. For this reason, we always assume that agents will be on polygonal vertices (even at the start of time) when they cannot perceive depth.

\section{Centralized Sequential and Parallel Algorithms} \label{sec:central}

We first present a centralized sequential algorithm and then briefly show how it can be parallelized. Our approach is to decompose the weak dual graph into a suitable set of at most $\lfloor n/2 \rfloor -1$ connected triplets and then assign a guard for every triplet. The high-level steps are outlined in Algorithm \ref{alg:central}. 
\begin{algorithm} %
\caption{Centralized algorithm for the connected art gallery problem.} %
\label{alg:central} %
\begin{algorithmic}[1]
    \STATE Compute a triangulation $T$ of the polygon $P$~\cite{C91} and then compute the weak dual graph.
    \STATE Decompose the weak dual graph into at most $\lfloor n/2 \rfloor -1$ triplets, i.e., groups of three connected nodes in $\T$, as described in Algorithm~\ref{alg:triplets}. \label{lno:decompose}
    \STATE Each triplet corresponds to three triangles arranged in such a way that there is a middle triangle that shares two edges, say $a$ and $b$, with the other two triangles. Placing a guard at the common vertex between $a$ and $b$ for every triplet is the required solution. \onlyLong{(See Figure~\ref{fig:triplets} for an illustration.)}
\end{algorithmic}
\end{algorithm}
\onlyLong{
\begin{figure}[h]
	\centering
		\includegraphics[width=0.5\textwidth,page=4,clip=true,trim=40 160 300 120]{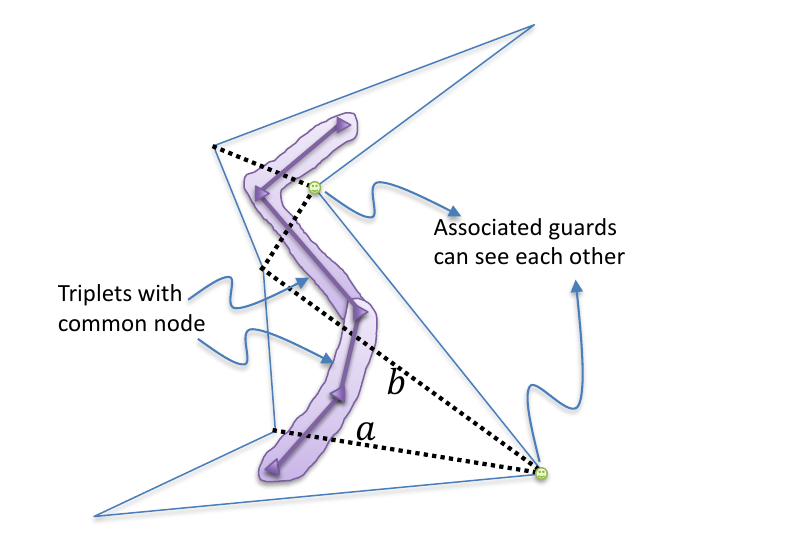}
	\caption{Illustrated placement of guards associated with triplets. Notice that the guards  are associated with the two triplets sharing a common node. Consequently, they can see each other.}
	\label{fig:triplets}
\end{figure}
}
Consider the weak dual graph $\T$, which is of course a tree with maximum degree 3. Root the tree at some node $r$ that is of degree 1.  A {\it triplet} is any set of three nodes in the tree that are connected. We now show a simple procedure (cf. Algorithm \ref{alg:triplets}) to decompose $\T$ into triplets. Subsequently, we will prove some properties of these triplets that will immediately  lead us to the required centralized algorithm for the connected art gallery problem.

\begin{lemma}
When two triplets share at least one common node, their associated guards can see each other.
\end{lemma}
\onlyLong{
\begin{proof}
Notice first that the guard in some triplet $(t_1, t_2, t_3)$ is placed in such a manner that for every $t_i$, $1 \le i \le 3$, the guard is one of the vertices of $t_i$. Thus, for any two triplets that share at least one common triangle $t$, the two guards associated with the two triplets must be on vertices of $t$. Thus they can see each other.  (See Figure~\ref{fig:triplets} for an illustration.)
\end{proof}
}

\begin{figure}[h]
	\centering
		\includegraphics[width=0.65\textwidth,page=3,clip=true,trim=0 210 0 30]{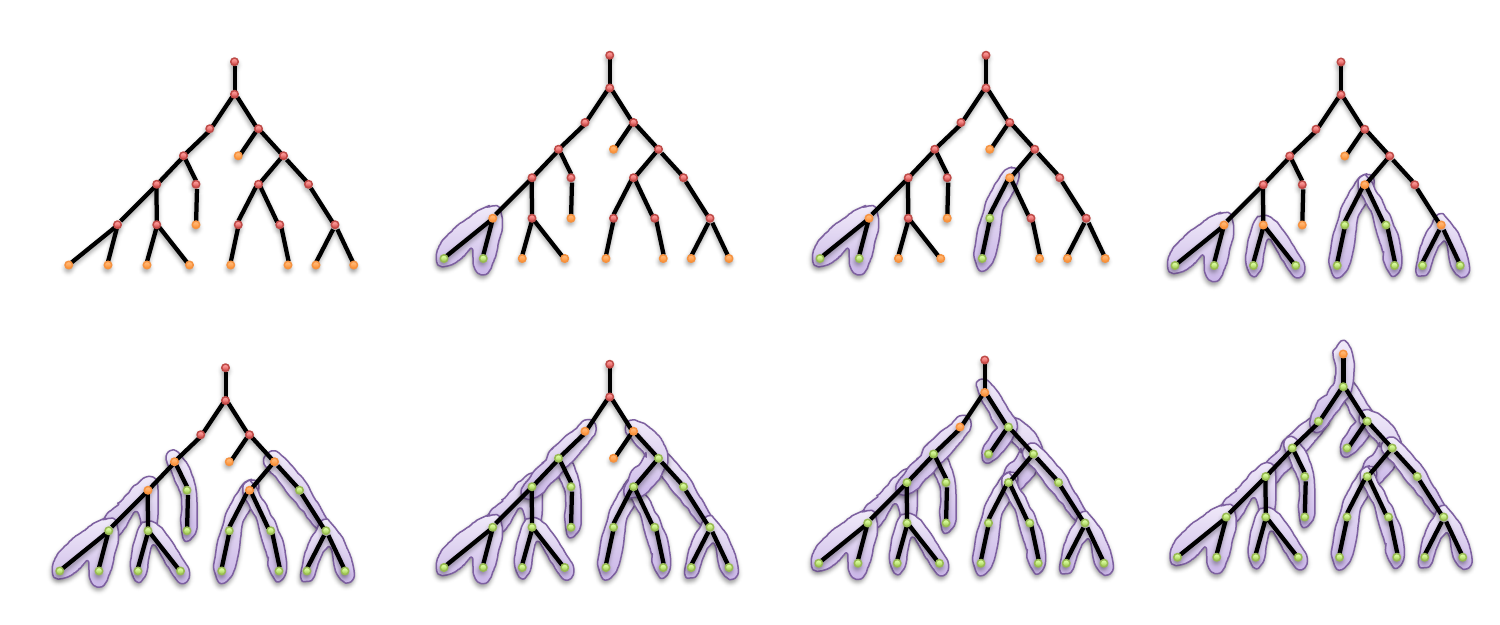}
	\caption{Sequence of states (mostly) at the end of each iteration of the {\tt for} loop in Algorithm~\ref{alg:triplets}.}
	\label{fig:decompose}
\end{figure}
\begin{algorithm} %
\caption{Algorithm to decompose $\T$ into triplets.} %
\label{alg:triplets} %
\begin{algorithmic}[1] %
    \REQUIRE A tree $\T$ rooted at a node $r$ of degree 1, with max degree three,  and depth $L$. Note that $r$ is assumed to be at level 0, so there are $L$ levels from $r$ to the farthest leaf (inclusive).
    \ENSURE A collection of triplets. 
\STATE Color all internal nodes red and all leaves orange.
\FOR{$\ell \leftarrow L$ down to $2$ (decrementing by 1 every iteration)}
	\WHILE{$\exists$ orange node $v$ at level $\ell$}
		\IF{$v$ has a sibling $v'$ that is also colored orange}
			\STATE Form a new triplet comprising $v$, $v'$, and their common parent  node $p$.
			\STATE Color parent $p$ orange.
			\STATE Color $v$ and $v'$ green.
		\ELSE
			\STATE Form a triplet comprising $v$, parent $p$ of $v$, and the grandparent $p'$ of $v$.
			\STATE Color $p'$ orange.
			\STATE Color $v$ and $p$ green.
		\ENDIF
	\ENDWHILE
\ENDFOR
\IF{the root is not part of some triplet}
\STATE Form a triplet comprising the root, its child, and an arbitrarily chosen grandchild. Color all three nodes green. \label{lno:last}
\ENDIF
\end{algorithmic}
\end{algorithm}

\onlyLong{
\begin{figure}[h]
	\centering
		\includegraphics[width=0.4\textwidth,page=8,clip=true,trim=200 160 270 120]{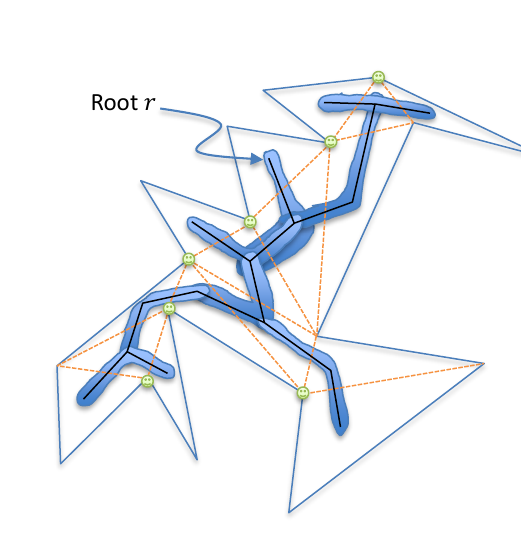}
	\caption{Illustrates the outcome of Algorithm~\ref{alg:central} along with the triangulation and the triplets.}
	\label{fig:connected}
\end{figure}
}

Having presented the algorithm to solve the connected art gallery problem in this centralized setting, we move on to analyze the algorithm. Our main focus will be on analyzing Algorithm~\ref{alg:triplets}. We make a series of observations formalized as lemmas and then derive the result as a consequence. For a given set of triplets, we define the {\it triplets graph} to be the graph with the triplets as vertices and edges between pairs of triplets that share at least one edge. We say that the set of triplets covers the tree $\T$ if every node is part of at least one triplet. \onlyShort{The proof of the following lemma is deferred to the full version.}

\begin{lemma}
Given $\mathsf{t}$ is the number of nodes in $\T$, we claim that Algorithm~\ref{alg:triplets}
\begin{enumerate}
	\item forms a set of triplets that covers $\T$,
	\item guarantees that the triplets graph	is connected,
	\item guarantees that the number of triplets formed is at most $\lfloor \mathsf{t}/2 \rfloor$, and
	\item runs in $O(\mathsf{t})$ time. %
\end{enumerate}
\end{lemma}
\onlyLong{
\begin{proof}
We address the statements in sequence.
\begin{enumerate}
	\item Coverage is obvious from the pseudo-code because all green nodes are covered by some triplet or the other and all nodes finish green.
\item To show connectedness, it suffices to show that every triplet is connected to the triplet $t^*$ that covers the root of $\T$. Suppose not. Let $t$ be the triplet that is not connected to $t^*$ such that its orange node $v$ (i.e., the node in $\T$ that was orange when $t$ was first formed) is closes to the root of $\T$. At some subsequent time, $v$ must be picked up form a triplet $t'$. Either $t'$ is not connected to $t^*$ (in which case $t$ is not the triplet with its orange node closes to the root) or $t'$ is connected to $t^*$ (in which case, so is $t$). Thus, either way, we reach a contradiction.

\item To bound the number of triplets, we use a counting argument for which we place half a dollar on each node of the tree. For each triplet other than the triplet added in line number~\ref{lno:last}, the algorithm must spend a dollar and it gets these from the lowest two nodes in the triplet and leaves the half dollar in the highest node intact (for a future triplet to consume). The invariant is that the green nodes have no money left, while the other nodes have their half dollar intact. 
Given the way we re-color nodes, it is clear that this invariant is maintained at least till the end of the {\tt for} loop. After the {\tt for} loop ends, the root is guaranteed to be a non-green color (therefore still has its half dollar), but the child of the root may be either orange or green. If it is green, then the triplet that first covered that node must have also covered the root and colored the root orange; so the triplets have covered the tree with a half dollar to spare. On the other hand, if the child of the root is orange, both it and the root have their half dollars intact and the last triplet (from line number~\ref{lno:last}) consumes them, thereby leaving us with a tree covered with triplets and no spare change left. From either case, we can conclude that the number of triplets is $\lfloor \mathsf{t}/2 \rfloor$.

\item The first line requires at most $O(\mathsf{t})$ time because triangulation has been shown to take at most $O(\mathsf{t})$ time~\cite{C91}. The last line also requires $O(\mathsf{t})$ time because there are $\mathsf{t}$ triangles in the triangulation. So we need to focus on the second line that deals with decomposing the tree into triplets. At each level, the time is at most proportional to the number of nodes in that level because for each node (that is colored orange) in that level, a short $O(1)$ time procedure is performed and then the node is never visited again.
\end{enumerate}
\end{proof}
}
Recalling that $\mathsf{t} = n-2$, we can conclude that the sequential algorithm runs in $O(n)$ time.

Finally, we remark that the algorithm described above can be implemented in parallel. In particular, consider the shared memory CREW PRAM model comprising $O(n)$ processors. Goodrich~\cite{G89} has already shown how to triangulate $P$ in $O(\log n)$ time under PRAM. We logically assign one processor per triangle and ensure the parent-child relationship between triangles is extended to the processors. Then, each iteration of the {\tt for} loop in Algorithm~\ref{fig:decompose} (comprising several {\tt while} loop iterations) can be executed in parallel. 
\begin{theorem}
Supported by Algorithm~\ref{alg:triplets}, Algorithm~\ref{alg:central} solves the connected art gallery problem with at most $\lfloor n/2 \rfloor -1$ guards in time that is linear in $n$. Moreover, can be solved in the CREW PRAM model with at most $\lfloor n/2 \rfloor -1$ guards in time that is linear in the diameter $D$ of the weak dual graph associated with the triangulation of $P$. 
\end{theorem}

\section{Distributed Guarding With Proximity Perception} \label{sec:algo1}

In this section, we consider the case where each agent is able to distinguish the proximity or relative distances (without knowing the actual distances) between the various objects associated with the polygon as well as with other agents, etc. which are in its visibility polygon at any specified moment. Specifically, the agents' {\tt "look"} paradigm is reflective of real-world sensing techniques where absolute distances to objects in the scene are unavailable whilst their relative distances can be inferred such as with photogramametric vision in drones. 

We give a distributed solution that solves the  connected art-gallery problem and runs in $O(\min(\tilde{d}_v^2, n))$ rounds, where $\tilde{d}_v$ is the minimal v-diameter that we formally define later in this section.  Our solution comprises two algorithms that are executed in parallel, one in a breadth first manner and the other in a depth first manner.  Our final solution is to take the best out of both explorations. \onlyLong{We will begin with the breadth first algorithm that runs in $O(\tilde{d}_v^2)$ rounds.} \onlyShort{Due to space limitation, we will describe the breadth first algorithm that runs in $O(\tilde{d}_v^2)$ rounds in detail. We subsequently present a brief overview of the depth first exploration and defer further details to the full version~\cite{fullversion}.\john{To be done.}} Finally, we conclude with some remarks on how our algorithm can form the basis for solving other polygon problems on $P$. 

Assuming that the vertices and the edges defining the polygon $P$ are in general position, the agents start at some vertex in $P$, which we can assume w.l.o.g. to be $p_1$.
The algorithm operates in phases.  %
At the end of a particular phase $\ell$, %
a subset $S_\ell$ of the agents have ``settle'' into their final positions while establishing a connected guarding of the subset $P_\ell$ of the polygon. For any agent $i$, its settled position is a vertex in $P$ and is denoted $s_i$. 

 Moreover, the settled agents are arranged in the following hierarchical manner. %
 W.l.o.g., let the root be agent 1, settled at $p_1$. We define the territory of the root, i.e.,  agent 1, to be $\territory(1) \triangleq \bar{V}_{p_1}^P$.   Every other settled agent  $j$  has a parent agent $\parent(j)$. If agent $i = \parent(j)$, then we say that $j$ is the child of $i$ denoted as $\child(i)$.  Each parent agent $i$ has one child agent $j$ per gap edge in its $\territory(i)$ and the child is located at one of the end points, say $p_a$, of a the gap  edge $(p_a, p_b)$. Thus $s_j = p_a$. The other end of the gap edge $p_b$ is denoted $\orient(j)$; intuitively, $j$ settles at $s_j$ and orients itself towards  $\orient(j)$ for performing ``look'' operations. 
 
 \noindent Further, $\territory(j) \triangleq \bar{V}_{s_j}^P \cap \crop(P,s_{\parent(j)}, s_j, \orient(j))$ i.e., it is the portion of $s_j$'s vertex limited visibility polygon not containing $s_j$'s parent and truncated by the gap edge that originated it. (See Figure \ref{fig:suboptimal}.) Intuitively, each agent $j$ is only responsible for guarding its $\territory(j)$. Notice that by definition, territories of a parent and its child do not overlap (they share a bordering edge that is a gap edge seen originally by the parent). Moreover, territories of children of a given parent do not overlap as well (at most they share a vertex) as this would imply that the polygon contains holes i.e., it is non-simple.
 Therefore, to ensure correctness, we must ensure that $\cup_j \territory(j) = P$.

Next, we define a {\it territory tree} to be a tree in which nodes are territories and edges are pairs of territories that share a common diagonal (gap) edge. Let $T^*$ be the set of all possible territory trees that can be achieved given all possible options for the starting vertex $p_1$ and all possible choice of placement of child agents. Then, we define $d$ as the maximum diameter of all such territory trees, i.e., $d \triangleq \max_{T \in T^*} \diam(T)$.

Initially, $S_0 = \{1\}$ (w.l.o.g.), $s_1 = p_1$, and $P_0$ is simply $\bar{V}_{s_1}^P$. We are now ready to present the steps to be performed within each phase $\ell$; notice that there cannot be more than $d$ phases to the algorithm hence, $1 \le \ell < d$. Intuitively, in each phase $\ell$ (see Algorithm~\ref{alg:suboptimal}), we incrementally construct the territories at level $\ell$ of the territory tree.
\begin{algorithm}
\caption{Phase $\ell \ge 1$ of the distributed algorithm for the %
connected art gallery problem that may use more than $\lfloor n/2\rfloor -1$ guards. This description assumes phases 0 to $\ell-1$ have completed and each agent $i \in S_{\ell-1}, \ell > 1$, remembers one marked vertex. (This marking scheme ensures that a child and grandparent are not visible to one another.)} %
\label{alg:suboptimal}
\begin{algorithmic}[1]
    \STATE Every settled agent $i \in S_{\ell -1}$ performs a {\it look} operation into its $\territory(i)$ and counts the number of gap edges in its $\territory(i)$. Call this count $b_i$. Each settled agent $i$ now up-casts $b_i$ to the root with intermediate settled agents aggregating the quantities by adding up the numbers sent by their children. 
	\STATE Notice that at the end of the up-casting, the root will know the total number $b$ of gap edges. The root apportions $b$ new agents and sends them to its children according to the numbers sent by each child. Subsequently, whenever a settled agent notices new agents reaching its position, it will apportion the agents according to numbers sent by its children and the new agents will move to their assigned child of the current settled agent. 
	\STATE Each settled agent $i \in S_{\ell-1}$  whose $\territory(i)$ has  some $b_i > 0$ gap edges gets exactly $b_i$ new agents. Agent $i$ assigns each of those new agents $j$ to an unmarked vertex of each such gap edge and consequently, agent $j$ marks the other vertex of that gap edge.
\end{algorithmic}
\end{algorithm}
\begin{figure}[h]
	\centering
		\includegraphics[width=0.5\textwidth,page=15,clip=true,trim=80 50 20 80]{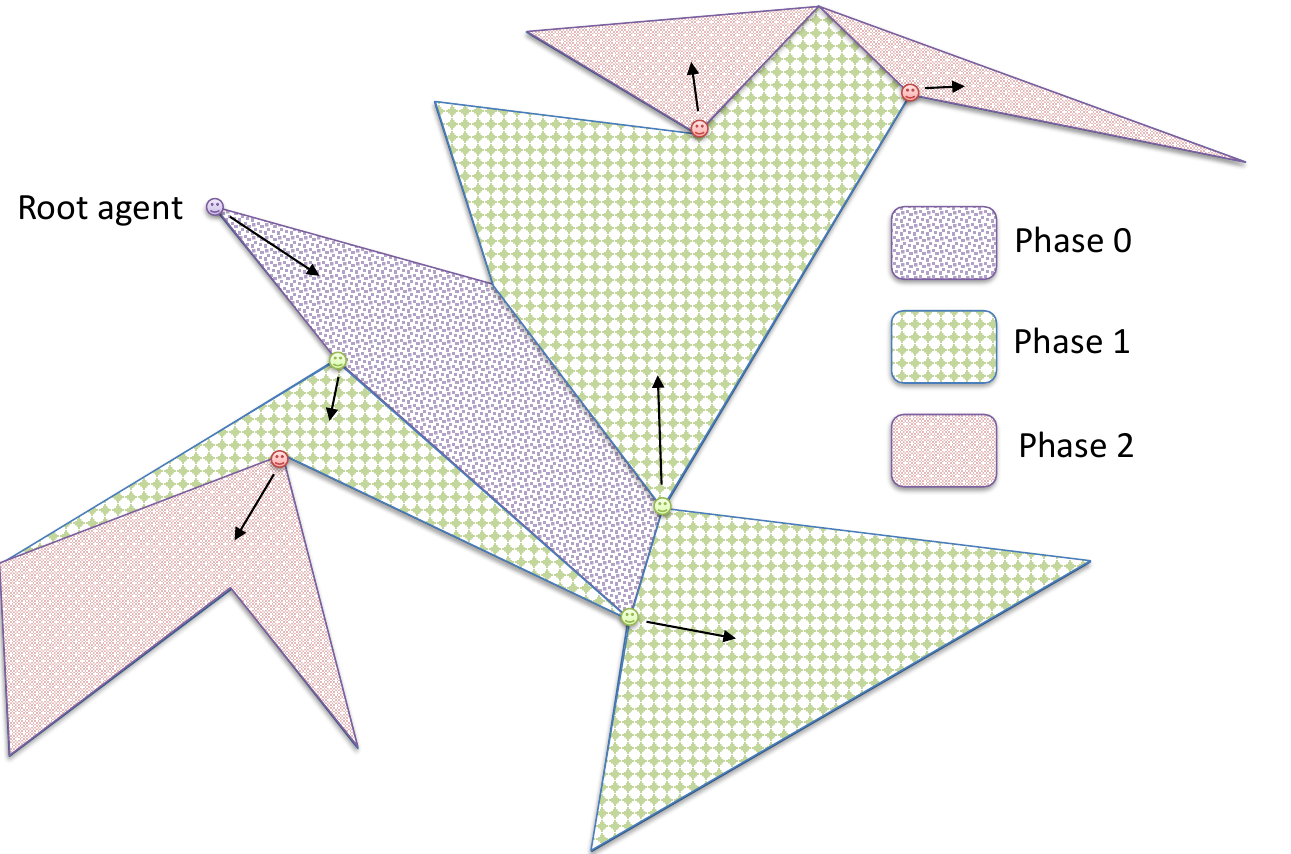}
	\caption{Depicts the placement of agents and their respective territories at three levels. The arrows point each agent into its territory.}
	\label{fig:suboptimal}
\end{figure}

\begin{lemma} \label{lem:nodepth}
Repeating Algorithm~\ref{alg:suboptimal} until all levels of the territory tree are explored, we get a distributed algorithm that, with no more than $n$ agents, ensures that the agents position themselves in a manner that solves the connected art gallery problem. The round complexity is $O(d^2)$.
\end{lemma}

\onlyLong{
\begin{proof}
To ensure correctness, we have to prove that (i) every point in $P$ is visible to at least one guard and (ii) the visibility graph of the guards is connected. Item (ii) is immediate from the fact that the new agents are always connected to their parents, so no isolated agents are ever created. Item (i) is also clear if we show that the territories form a partitioning of the polygon -- to this end, we have already seen previously that the set of territories generated by running Algorithm \ref{alg:suboptimal} are disjoint and each territory is guarded by at least one of the agents as ensured by the deployment step of agents in each phase. Now, to show that the algorithm terminates we observe that in each phase, every territory, apart from the root, is uniquely associated with a gap-edge (namely its originating gap-edge) and moreover, two different territories do not share this originating gap-edge. Since the total set of candidate gap-edges is finite ($< \binom{n}{2}$), it must be that the algorithm terminates eventually. Moreover, if the algorithm terminates such that $\cup_j \territory(j) \subsetneq P$, then $P - \cup_j \territory(j)$ and $\cup_j \territory(j)$ share an internal diagonal, say $\overline{p_i p_j}$, of the polygon which belongs to the territory, say $\territory(q)$ -- this means that during the round just after deployment of agent $q$, $\overline{p_i p_j}$ will be detected as a gap-edge, causing at least some portion of the un-guarded region to become guarded within the subsequent phase - a contradiction! Thus, when the algorithm terminates it must be that $\cup_j \territory(j) = P$.

The running time comes from the fact that each time Algorithm~\ref{alg:suboptimal} is invoked, it requires at most $O(d)$ time to complete because the main time consuming task is the up-casting of $b_i$ values and the down-casting of the required number of agents. Since the depth is at most $d$, Algorithm~\ref{alg:suboptimal} is invoked at most $d$ times and we get the required result.
\end{proof}
}

Here, we introduce the notion of {\it minimal visibility connected} {\it vertex guarding} (henceforth referred to as {\it minimal v-guarding}) which is pivotal in this case for developing algorithmic bounds on the running time. 
Let $P$ be a simple polygon with \(n\) vertices. Recall that,
given a set of labelled guards \(G=\{g_1, g_2, \ldots, g_k\}\)
of $P$, we associate with $G$ a unique graph \G\ with a vertex set of size $k$  \onlyLong{which abstracts the visibility relations between the guards w.r.t.
the polygon $P$:} such that when two guards are visible to each other, then they are connected by an edge in this graph \G\, i.e.,
\(e=\{i,j\} \in E[\G] \iff g_i\) is visible to \(g_j\). We say that \(G\) is a {\it minimal v-guarding}
or a {\it minimal v-configuration} of $P$ whenever the
following holds, %
\(\forall v \in \G,\) \textbf{{\it at least}} one of the following two
conditions applies:
\begin{enumerate}
\item
  \(\G - v\) has more than a single component.
\item
  the vertex guarding \(G_{-v}:= G \setminus \{g_v\}\) of
  $P$ is incomplete, i.e.,
  \(\cup_{g \in G_{-v}} \partial V_g \subsetneq \partial P\) where $\partial P$ is the set of vertices of polygon $P$ and  $\partial V_g$ refers to the set of vertices of $P$ visible from $g$\onlyLong{ (set of nodes in the visibility graph $V_g^P$)}.
\end{enumerate}

We now define the notion of {\it minimal v-diameter}
\(\tilde{d}_v := \max\limits_{\tau \in \mathcal{M}} {\sf diameter}(G_\tau)\)
where \(\mathcal{M}\) is the set of all possible minimal v-configurations of $P$
and \(G_\tau\) denotes the associated visibility graph of such a
guarding \(\tau\) from $\mathcal{M}$.

In order that we may successfully compare the efficiency of our algorithms with one another we use $\tilde{d}_v$ which is a polygonal parameter pertinent to our guarding problem.

\begin{lemma} \label{lem:twodiams}
$d = O(\tilde{d}_v)$ i.e., the diameter of any territory tree on $P$ is asymptotically bounded by the minimal v-diameter of the Polygon $P$.
\end{lemma}
\onlyLong{
\begin{proof}
Our strategy to prove lemma \ref{lem:twodiams} requires us to show that any configuration of the agents generated by our algorithm, henceforth called a solution-configuration (i.e., as arrived at by running our algorithm) can be put in correspondence with a minimal v-configuration such that the diameter of said minimal v-configuration is no less than that of the solution-configuration in consideration. Without loss of generality, assume that the guards and vertices of the polygon are in general position. The strategy is as follows:

\begin{itemize}
   \item Consider the territory tree {\tt tt} associated with the solution-configuration, isolate the set of agents on the longest path of {\tt tt} from root to any leaf: call it $L_0$.
    \item Observe that the visibility graph of $L_0$ is minimally connected since no three nodes in $L_0$ forms a triangle in the visibility graph. \footnote{The only possibility is that \{j, \parent(j), \child(j)\} forms a triangle: since the only way a \child(j) could see \parent(j) is if the bot $j$ during its turn in run of Algorithm \ref{alg:suboptimal} locates a gap-edge that is once again incident on \orient(j) and places the bot on this vertex leading to \child(j) = \orient(j) - however, recall that the marking scheme precisely avoids this placement by forcing $j$ to choose the other vertex of the gap-edge for deploying the bot, thereby avoiding this visibility triangle formation.} Moreover,\\ {\sf diameter}($L_0$) = $\Theta$({\sf diameter}({\tt tt})).
    \item Transform $L_0 \xrightarrow[\text{extend}]{\text{minimal}} G_L$ i.e., starting from the partial solution-configuration $L_0$, we obtain $G_L$ which is a minimal v-configuration (Construction described in algorithm \ref{alg:minxtnd}).
    \item Applying Lemma \ref{lem:minxtnd} inductively starting from $L_0$, we have that {\sf diameter}($G_L$) $\geq$ {\sf diameter}($L_0$).
\end{itemize}
\end{proof}

\begin{algorithm}
\caption{Minimal-extension of $L_0$ to $G_L$. %
W.l.o.g. we assume that all the guards, vertices and edges are in general position. Note: A {\sf v-edge} of $P_j$ is an edge with at least one end point being a vertex of $P$ and {\tt cl(R)} denotes the set-closure of region R under the usual $\ell_2$-norm.}
\label{alg:minxtnd}
\begin{algorithmic}[1]
    \STATE Let $L_0 = \{g_1, g_2, \ldots, g_k\}$ and set $j$ to 1.
    \STATE set $P_j = {\tt cl}(P \setminus \bigcup\limits_{h \in L_{j-1}}V_g)$ and $\partial P_j = \{v \in P_j : v\ is\ a\ vertex\ of\ P \}$ 
    \STATE If $\partial P_j \neq \phi$ then, {\tt choose} $e_j$, a {\sf v-edge} of $P_j$: we know that $e_j$ is part of the boundary of {\it exactly one} visibility polygon (refer Lemma \ref{lem:minxtnd}), say of $g \in L_{j-1}$; {\tt else} we terminate the algorithm and output $G_L = L_{j-1}$.
    \STATE Let $L'_{j-1} = L_{j-1} \setminus \{g\}$. Set $P'_j =  P_j \setminus {\tt cl}\big(P \setminus \bigcup\limits_{h \in L'_{j-1}}V_g\big)$
    \STATE Pick a point $p'_j \in P'_j$, in general position, such that $L_j = L_{j-1} \cup \{p'_j\}$ is a (partial) minimal v-configuration.
    \STATE Increment $j$ and repeat from step 2.
    
\end{algorithmic}
\end{algorithm}

Note that the algorithm \ref{alg:minxtnd} is used just as an analysis and the actual running time is irrelevant. 

\begin{lemma} \label{lem:minxtnd}
In Algorithm \ref{alg:minxtnd}, whenever $L_{j-1}$ is a partial minimal v-configuration of polygon $P$ such that, in step 3, $\partial P_j \neq \emptyset$, then there exists a {\sf v-edge} $e_j$ of $P_j$ such that there is a unique $g \in L_{j-1}$ such that $e_j \in V_{g}$. Moreover, there always exists at least one point $p'_j \in P'_j$ (in general position w.r.t. $L_{j-1} \cup \partial P$) such that
\begin{enumerate}
    \item  $L_j \triangleq L_{j-1} \cup \{p'_j\}$ is a (partial) minimal v-configuration of $P$ and
    \item The set of guards $L_j$ strictly guard more vertices of $P$ than $L_{j-1}$ i.e.,
    \[\bigcup\limits_{h \in L_{j-1}}\partial V_h \subsetneq \bigcup\limits_{h \in L_{j}}\partial V_h\]
\end{enumerate}
\end{lemma}

\begin{proof}
$\partial P_j  \neq \emptyset \implies P_j \neq \emptyset$, therefore, $\bigcup\limits_{g \in L_{j-1}}V_g$ is a strict sub-polygon of $P$ and so must have an $e_j$ which is a {\sf v-edge}. This $e_j$ is part of at least one visibility polygon, say of $g \in L_{j-1}$. Since all the guards and vertices i.e.,  $L_{j-1} \cup \partial P$ are in general position and a {\sf v-edge} contains at least one vertex of $P$, this gives us that $g$ is the unique guard that gives us $e_j$. By definition, $P'_j$ is the portion of $P$ that is visible only to $g$ among $L_{j-1}$. Now, consider the piece of $P_j$ that is bounded by $e_j$, let us triangulate it. Let $v$ be the vertex of P that is part of the triangle bounded by edge $e_j$ but not incident to $e_j$. Consider the angle bisector $l_v$ of $v$ w.r.t. this triangle of the angle opposite $e_j$ - extend it beyond $e_j$ into $P'_j$. Now, this portion of the bisector inside $P'_j$ has infinitely many candidate points for $p'_j$ such that $v$ is visible to it and by virtue of being in $P'_j$ is visible to exactly $g$ among $L_{j-1}$! Moreover, since $v$ is visible to this $p'_j$, we have $v \in \bigcup\limits_{h \in L_{j}}\partial V_h \setminus \bigcup\limits_{h \in L_{j-1}}\partial V_h$ completing the proof.
\end{proof}
}
Thus, Lemma \ref{lem:nodepth} along  with Lemma \ref{lem:twodiams} yields the following theorem:
\begin{theorem}
\label{thm:tilde_dv_alg}
There exist a distributed algorithm that solves the connected art gallery problem in $\mathcal{O}(\tilde{d}_v^2)$ time using no more than $n$ agents limited to proximity sensing capability, where $\tilde{d}_v$ refers to the minimal v-diameter of $P$ and $n$ is the number of polygon vertices.
\end{theorem}

\noindent{\bf An $O(n)$ round Algorithm.}
While most real-world polygons may have small diameters, it is nevertheless conceivable that $\tilde{d}_v \in \Omega(n)$ in some cases like spirals. The algorithm that we have presented above will unfortunately require quadratic in $n$ number of rounds for such situations, which is  undesirable. \onlyShort{In the full version, we describe a depth first procedure that only requires $O(n)$ rounds.}
\onlyLong{
So to mitigate such situations, we now sketch a simple depth first procedure that runs in $O(n)$ rounds. We assume, as before, that $n$ agents start at $p_1$. Initially, they are all exploring agents and all exploring agents stay together. As the algorithm progresses, the agents (one-by-one) turn into settled agents that don't move. As before, agent 1 settles in $p_1$ (thus, it's now a settled agent and not an exploring agent) and takes charge of guarding $\territory(1) \triangleq \bar{V}_{p_1}^P$. The gap edges in $\territory(1)$ are numbered  in some well-defined order, say, the clockwise order. Assuming $\territory(1)$ has gap edges, the remaining $n-1$ unsettled agents move simultaneously to an arbitrarily chosen end point  (say $p_a$) of the  first gap edge (say $(p_a, p_b)$); agent 1 remembers $(p_a, p_b)$ as visited. Agent 2 settles at $p_a$ and guards $\territory(2) \triangleq \bar{V}_{p_a}^P \cap \crop(P,p_1, p_a, p_b)$; agent 2 remembers agent 1 as its parent. If $\territory(2)$ has gap edges, the remaining $n-2$ exploring agents repeat the process and explore the first gap edge in $\territory(2)$. Whenever an agent $j$ settles in an endpoint $p_x$ of some gap edge $(p_x,p_y)$, it guards $\territory(j) \triangleq \bar{V}_{p_x}^P \cap \crop(P,s_{\parent(j)}, p_x, p_y)$, where $s_{\parent(j)}$ denotes the location of the parent of agent $j$. Moreover, agent $j$ must keep track of the gap edges in $\territory(j)$ that have not been explored yet.  If $\territory(j)$ does not have gap edges, the exploring agents move to agent $j$'s parent and explore unvisited gap edges associated with the territory of the parent. When all gap edges associated with $\territory(j)$ have been explored, the remaining exploring agents move to agent $j$'s parent and explore unvisited gap edges (if any). One can observe immediately that the process is depth first and the graph on territories with graph-edges between two territories separated by a gap edge will form a tree because in our depth first exploration, we will not encounter the equivalent of a back-edge as that would imply a hole in the polygon. }

\begin{remark}
Both the $O(\tilde{d}_v^2)$ algorithm and the depth first $O(n)$ algorithm will maintain agents connected by line of sight. So we can start both algorithms simultaneously and, when one of them -- the winner -- finishes, the other can be terminated by a special terminate message that will take at most $O(\min(d_v, n))$ rounds to reach all. 
\end{remark}

\subsubsection*{General Problem Solving Given a Visibly Connected Guard Placement.} 
We have now provided two different algorithms that take, respectively, $O(\tilde{d}_v^2)$ rounds and $O(n)$ rounds. Throughout the course of the algorithms, the agents stay connected through visibility links. So with sufficient agents (at most $O(n)$), we can execute both algorithms simultaneously. When one of them terminates and the root is aware of the termination, the other algorithm can be terminated prematurely along with a request for all agents in the prematurely terminated algorithm to collect at the root. This will only require $O(\tilde{d}_v)$ rounds; each such agent $a$ waits for all its children to reach its position and they can then collectively move to the parent of $a$. Thus, we only require $O(\min(\tilde{d}_v^2, n))$ rounds in total.
\enlargethispage{2\baselineskip}
Moreover, since the agents are settled into a visibly connected guarding position and are aware of their respective territories in the territory tree, they can gather at the root's position in a bottom-up fashion. Thus, the root, in $O(\tilde(d)_v)$ rounds, can collect the views of all the agents and perform computation using  the collective views of the agents. Thus, we get the following generalized theorem.

\begin{theorem}
Suppose ${\cal P}$ is a computationally tractable problem that takes a polygon $P$ as input and either 
\begin{itemize}
\item  outputs information in the form of bits
\item or requires placing agents in positions within $P$.
\end{itemize} 
Then, ${\cal P}$ can be solved in our distributed context in $O(\min(\tilde{d}_v^2, n))$ rounds. 
\end{theorem}

\section{Distributed Guarding With  Distance Perception} \label{sec:algo2}
\enlargethispage{2\baselineskip}
In this section, we give a distributed algorithm that solves the connected art gallery problem in $O(D^2)$ rounds, where $D$ is the (unweighted) diameter of the medial axis. This algorithm's advantage is that its running time depends on $D$, which is more well-known than $\tilde{d}_v$. However, as opposed to the previous section, in this case, agents require the ability to perceive depth. \onlyShort{The key idea of the algorithm here is if agents were placed on all internal nodes of the medial axis and some specially chosen vertices, they cover the entire graph as well as remain visibly connected.}%

\onlyLong{The key idea of the algorithm here is if agents were placed on all internal nodes of the medial axis and some specially chosen vertices, they cover the entire graph as well as remain visibly connected. We describe briefly regarding the procedure in which the algorithm computes adjacent nodes in the medial axis and how the algorithm assigns/places agents, and later describe the pseudo-code.}

\para{Computing adjacent nodes in the medial axis} Imagine an agent at any point $x$ on the medial axis. The agent can simulate the creation of a maximal disc at $x$ to find the objects (vertices or edges) that determine $x$, i.e., the objects due to which $x$ is a part of the medial axis. There would be at least two such objects that determine $x$. If $x$ is determined by multiple objects ($>2$), it implies that $x$ itself is a node on the medial axis, and we can consider any two consecutive objects determined by the \textbf{look} operation. For example, if $a,b,c,d$ are 4 objects, that determine $x$ and are ordered in accordance with the look operation, we consider the pairs $ab, bc, cd$ and $da$ only. Note that, only the consecutive object pairs determine the medial axis edges incident at $x$, and hence only those are considered.

For each pair of objects $ob_1$ and $ob_2$, there can only be three possible cases; either both are vertices, both are edges, or one of them is a vertex while the other is an edge. For all the cases, the agent at $x$ is aware of the structure of the medial axis from $x$. Thus, if both $ob_1$ and $ob_2$ are vertices, then the next node of the medial axis lies on the perpendicular bisector of the line segment $(ob_1,ob_2)$. If both $ob_1$ and $ob_2$ are edges, then the next node of the medial axis lies on the angle bisector of $ob_1$ and $ob_2$. Lastly, w.l.o.g. if $ob_1$ is a vertex and $ob_2$ is an edge, then the next node of the medial axis lies on the parabola determined by $ob_1$ and $ob_2$. Since agents have infinite computing power and depth sensing ability, they can progressively simulate maximal discs along the medial axis structure (perpendicular bisector, angle bisector or parabola) until the maximal disc encounters a new object (say $ob_3$). The center of the maximal disc at this instance determines the next adjacent node in the medial axis. 
We define the set of new adjacent nodes obtained in phase $i$ as $A_i$.

\para{Agent placement} As in Algorithm \ref{alg:suboptimal}, when an already placed agent $a$ determines its adjacent set of positions on which new agents are to be placed, then $a$ upcasts the request of the required number of agents up to the root (the spot initially containing all the agents) with intermediate agents aggregating the quantities by adding up the numbers sent by their children. The root serves the request by assigning the required number of agents. The assigned agents trace back the path to $a$ and thereafter get placed in their determined spot. 

\enlargethispage{\baselineskip}
\small
\begin{algorithm} %
\caption{An $O(D^2)$ time algorithm for the connected art gallery problem.} %
\label{alg:medial} %
\begin{algorithmic}[1]
\STATE Starting from the initial given vertex $v$ where all the agents are placed, a medial axis point $m$ is determined. 
If $v$ is convex, then $m$ is given by $v$'s adjacent node in the medial axis. %
 Alternatively, if $v$ is a reflex vertex, pick the nearest visible new object $ob_3$ (determined by the depth sensing ability of the agents) not including $v$ and choose the center of the maximal disc determined $v$ and $ob_3$ as the point $m$ on the medial axis. 
\STATE Consider \onlyLong{the determined medial axis point} $m$ as the root. All agents are moved here.
\STATE Determine all the adjacent medial axis nodes from $m$ (i.e., the set $A_1$). \textit{$\backslash *$ This marks the end of the first phase. Each iteration of the loop represents a subsequent phase. The algorithm continues until the entire medial axis is uncovered. $*\backslash$}
     \FOR {each new adjacent node $x \in A_{i-1}$ determined in the previous phase, that is not a leaf node of the medial axis tree}
        \STATE {Compute set $A_i$ (current set of new adjacent medial axis nodes of $x$) in parallel.}
        \STATE {Place an agent at each node $y \in A_i$ except when $y$ corresponds to a convex vertex of the polygon (leaf node of the medial axis tree).}
        	\IF {$y$ is a part of a parabola determined by a polygon vertex and an edge and the polygon vertex does have an agent on it}
        	\STATE {Place an agent on the reflex vertex determining the parabola.}
        	\ENDIF
    \ENDFOR
\end{algorithmic}
\end{algorithm}
\normalsize

\begin{lemma} \label{lem:vcp}
Algorithm \ref{alg:medial} gives a visibly connected guard placement while ensuring that the entire polygon is guarded/covered.
\end{lemma}

\onlyLong{
\begin{proof}
To show that the agents guard the entire polygon, we consider a decomposition of the polygon inspired by the classical domain decomposition lemma \cite{Choi97mathematicaltheory}. The proof follows by showing that, each decomposed part is fully guarded by the agents and the union of the decomposed parts determines the entire polygon.

Recall that the medial axis tree of a polygon consists of both line segments and parabolas. By a slight abuse of notation, we consider an edge of the medial axis to either be a line segment or a parabola determined by its end points. %
Consider any non-leaf edge of the medial axis tree. We create the decomposition induced by the medial axis edges as follows. %
A decomposed part is a sub-division of the polygon bounded by the objects that determine the medial axis edge (including additional objects that might just determine the medial axis edges' end points) and the maximal discs centered at both end points of the medial axis edge.
For leaf edges (edges containing a convex polygon vertex), the decomposed part is bounded by the polygon edges that meet at the convex vertex along with the maximal disc centered at the adjacent medial axis node (other node of the medial axis leaf edge). It is easy to see that each decomposed part is a closed figure.
It is also to be noted, since the medial axis is a connected tree, the maximal discs centered at all internal medial axis nodes (non-leaf nodes) are a part of at least 2 adjacent decomposed parts. Consequently, the adjacent decomposed part begins exactly at the polygonal edge termination point of the previous decomposed part. 

Note that, consecutive decomposed parts are continual polygonal parts with overlapping maximal discs at the ends (with no missing portions). Since the medial axis spans the entire polygon, and the decomposition happens along the medial axis edges, it is not difficult to visualize that the union of all decomposed parts indeed gives us the entire polygon.

To show that the agents are visibly connected, we rely on the structure of the medial axis tree. Clearly, visible connectivity is maintained across all edges of the medial axis that are line-segments. If all edges of the medial axis are line segments, visible connectivity immediately follows from the connected tree structure of the medial axis. The only condition that visible connectivity might be lost is due to the existence of parabolas in the medial axis. We preserve the visible connectivity by placing an agent at each vertex that determines a parabola on the medial axis (see Algorithm \ref{alg:medial}). Consider any parabolic edge $(a,b)$ on the medial axis determined by a vertex $v$ and a polygon edge $l$. Since there exists a maximal disc centered at $a$ (resp. $b$), that touches $v$ (by the property of the medial axis), it implies that $v$ is visible from $a$ (resp. from $b$). In the presence of a polygonal boundary blocking visibility, such a maximal disc would not have been possible. This shows that visible connectivity is maintained across all parabolic edges. The overall visible connectivity follows from the connected tree structure of the medial axis and the fact that visible connectivity is maintained across all medial axis edges.
\end{proof}
}

\begin{theorem}
There exists a distributed algorithm that solves the connected art gallery problem in $\mathcal{O}(D^2)$ time using no more than $n$ agents, where $D$ refers to the medial axis diameter and $n$ is the number of polygon vertices.
\end{theorem}
\onlyLong{
\begin{proof}
The correctness of the algorithm follows directly from Lemma \ref{lem:vcp}. The time complexity is determined by the phases of the algorithm.  \onlyLong{Each phase takes at most $O(D)$ rounds and the number of phases is at most $O(D)$, thereby giving us an $O(D^2)$ algorithm. 

}
Starting from a medial axis point $m$, in each phase, agents get placed on all the adjacent medial axis nodes (for the case of parabolic path, an additional agent gets placed at the reflex vertex determining the parabola). As the diameter of the medial axis is $D$, there can be at most $O(D)$ phases. Additionally, in a phase, the agent placement procedure can request new agents from the root; this upcast and response takes place over the medial axis and can take up to $O(D)$ rounds. This validates the previously stated statement of having at most $O(D)$ phases with each phase taking at most $O(D)$ rounds, thereby giving us an $O(D^2)$ algorithm. 

Next, we show that the number of agents required is $\leq n$. Note that agents are only placed on the initially determined medial axis point $m$ or on the internal medial axis nodes (the non-leaf nodes of the medial axis) or on reflex vertices that determine a parabolic edge of the medial axis. 
In the given polygon $P$, let it contain $c$ convex nodes and $r$ reflex nodes, i.e., $c+r =n$. %
Observe that, any convex vertex of the polygon $P$ is a leaf in the medial axis tree of $P$. Conversely, any leaf of the medial axis tree of $P$ is also a convex vertex of $P$. This implies that the medial axis has $c$ leaves. Since the medial axis is a tree, the maximum number of internal nodes can be at most $c-1$. %
We consider the worst case, where agents are placed on all $r$ reflex vertices. Thus, the total number of agents placed equals $1+(c-1)+r$ which is $\leq n$. This completes the theorem proof.
\end{proof}
}
To reduce the final number of guards placed, we use similar procedure as described in Section \ref{sec:algo1}. This gives us the following theorem.

\begin{theorem} \label{thm:nodepth}
There exists a distributed algorithm that uses fewer than $n$ agents to compute the placement of at most $\lfloor n/2 \rfloor -1$ guard agents in a visibly connected manner, when the agents have depth sensing ability. Moreover, this algorithm takes at most $O(D^2)$ communication rounds.
\end{theorem}

\section{Lower Bound} \label{sec:lb}
\enlargethispage{\baselineskip}
In this section, we give lower bounds for a slightly weaker polygon exploration problem that requires for every point in $P$, that some agent must have been within line of sight of that point at some time instant during the course of the algorithm. Clearly, any solution to the visibly connected guard placement problem will also be a solution for the exploration problem. The lower bounds highlight the criticality of parameters like the medial axis diameter $D$ and the minimal v-diameter $\tilde{d}_v$ for solving the connected art-gallery problem.
The main result is summarised by the following Theorem.

\begin{theorem}
\label{thm:lb}
For every deterministic distributed guard placement algorithm $A$ with distance perception (resp., proximity perception)
there exists a polygon $P$ with medial axis diameter $D \in o(\log n)$ (resp., with minimal v-diameter $\tilde{d}_v \in o(\log n)$) such that $A$ requires $\Omega(D^2)$ time
(resp., $\Omega(\tilde{d}_v^2)$ time) 
to place the guards even when $A$ is provisioned with a number of guards that is $\Theta(n)$.
\end{theorem}

\noindent Our strategy is to show a reduction, specifically, we will reduce any instance of the well-studied tree exploration game \cite{DISSER2018} to the problem of placing guards in an unexplored polygon. \onlyShort{We briefly sketch our approach and defer details to the full version of the paper.}
\begin{enumerate}
    \item Firstly, we embed the tree in the Euclidean plane such that no two edges overlap excepting at a shared node and no two adjacent edges form an 180 degree angle.
    \item We {\it thicken} the edges of the embedded tree and consider the boundary of the union of the thickened edges to form a simple polygon.
    \item With the embedding and thickening transformations, we map the problem of collaborative exploration of the underlying tree to the guard placement in the obtained polygon. %
\end{enumerate}

\onlyLong{
\begin{lemma} \label{lb:embed} Given a complete $\Delta$-ary tree \T(V,E) with height $h$, there exists an embedding, \embed{\T}, in which no two {\it adjacent} tree edges are at an angle of $180$ degrees %
and no two tree edges touch each other, except when they share a common vertex.  %
\end{lemma}
\begin{proof} Without loss of generality, assume that $\Delta$ is even; Let $\Delta' = \Delta+1$. For $0\leq j \leq h$, let $N_j$ be the set of all nodes of the tree at depth $j$ - so $N_0$ is the singleton set containing the root and $N_h$ is the set of all leaves of \T\ where \T\ is the complete $\Delta'$-ary tree of depth $h$. Clearly, $\{N_j\}$ is a partition of the set of vertices. Additionally, assume that for each $N_j$, nodes are labelled in consecutive order (inductively), i.e., if $v_j(k) \in N_j$, then its parent is the node labelled $v_{j-1}({\big\lfloor \frac{k-1}{\Delta'} \big\rfloor+1})$ and its children are exactly the nodes of the form $v_{j+1}(\Delta'(k-1)+l), l \in \{1, \ldots, \Delta'\}$. 

Our strategy is to arrange all the nodes in concentric circles so that nodes of a given depth all lie on the same circle. The edges embedded are just the line segments connecting a parent to its child. We first create a $(\Delta+1)$-ary tree and then remove some vertices such that we have a $\Delta$-ary tree with the required properties. We inductively arrange the children of a given parent, then, we remove one out of the $\Delta'$ children. To be precise, we remove the child which, if present, is part of the edge that is closest (and in certain instances possibly equal) to forming 180 degrees (or $\pi$ radians) with the grandparent-parent edge (Notice here that there is precisely one child, if at all, which can be collinear with grandparent and parent). Thus, at each stage, since we remove exactly one child per parent off the $\Delta'$-ary tree, we end up finally with a $\Delta$-ary tree of depth $h$.
In detail, we embed \T\ inductively in the Euclidean plane as follows:

Firstly, we place the root $r$ of the tree at the origin. Next, place the $\Delta'$ nodes of $N_1$ equally spaced on the circumference of the first circle in clockwise order. Note that no two vertices of $N_1$ lie on the same diameter since $\Delta'$ is odd. Now, assume that the nodes in $\bigcup_0^q N_j, q \geq 1$ have all been placed. Thus, the nodes of up to depth $q$ all lie on or within some circle - let its radius be $R_q$. We now position the set of nodes in $N_{q+1}$ on the circle $\mathcal{C}_{q+1}$ centered on the root vertex and having radius $R_{q+1} > R_q$. Consider an embedded node $x = v_q(k)$ where $k \in \{1, 2, \ldots, (\Delta')^q\}$. Let its immediate clockwise neighbor in the embedding be $y$. Consider the portion of $\mathcal{C}_{q+1}$ within the conic region formed by the rays $\overrightarrow{rx}$ and $\overrightarrow{ry}$. Let $z = \mathcal{C}_{q+1} \cap \overrightarrow{rx}$ and $w = \mathcal{C}_{q+1} \cap \overrightarrow{ry}$. We partition the arc $\widehat{z w}$ of circle $\mathcal{C}_{q+1}$ into $\Delta'$ smaller arcs, namely, $\widehat{z z_1}, \widehat{z_1 z_2}, \ldots, \widehat{z_{\Delta'-1} z_{\Delta'}}$ where $z_{\Delta'} = w$. We identify the child $v_{q+1}(\Delta'(k-1)+l)$ with the point $z_l$ where $l \in \{1, \ldots, \Delta'\}$.

Note that the edges $xz_l$ and $xz_m$ cannot be parallel by this construction. In fact, the only potential for two adjacent edges (sharing vertex $x$) to be parallel is if $g = v_{q-1}({\big\lfloor \frac{k-1}{\Delta'} \big\rfloor+1}), x, z_l$ are collinear for some $l \in \{1, \ldots, \Delta'\}$. However, such a scenario of collinearity with $g, x$ is true for at most one value of $l \in \{1,2, \ldots, \Delta'\}$ - we ensure the non-parallel nature of adjacent edges for the partial embedding obtained so far by deleting the vertex $z_l$ which lies closest to the point $\overrightarrow{gx} \cap \widehat{zw}$ (and first to appear along $\widehat{zw}$ moving in the clockwise direction). Thus, by proceeding inductively, we obtain a planar embedding of a $\Delta$-ary tree of depth $h$.
\end{proof}

Now that we know how to embed \T, we construct a polygon $P$ from this euclidean embedding as follows:
Let \thick(\T) denote the set of all points that are strictly within $\epsilon$ distance of any point of \embed{\T}. %
By replacing all the circular arcs in the boundary (i.e., the closure of \thick(\T)) by line segments joining the endpoints of the arcs, we obtain a polygon $\mathcal{P}$ for a suitably chosen small enough value of $\epsilon$. The vertices of Tree \T\ correspond to convex {\it chambers} in polygon $P$ and correspondingly, the edges of \T\ are transformed into rectangular {\it corridors} by construction. %

We provide a reduction for the problem of solving the collaborative tree exploration \cite{DISSER2018} into one of guarding of a polygon via exploration in the distributed setting. %
Notice that due to the particular embedding of \T\ and the thickening procedure, the tree \T\ intuitively replicates the medial axis of the polygon.
(This is taken care by the properties of \embed{\T} shown in Lemma \ref{lb:embed}.) Joining the tree nodes with the required convex vertices in a proper fashion, exactly determines the medial axis. As such the diameter of the tree \T\ (represented by $d_{\T}$) differs from $P$'s medial axis diameter by a small constant (at most 2). Thus,  $d_{\T} \in \Theta(D)$.

Each {\it chamber} and {\it corridor} requires at most one agent in them as far as vertex guarding is concerned. Moreover, if, say $k \geq 1$ guards are placed in any one of these regions (to ensure overall visibility graph connectivity moving from {\it chamber} to {\it corridor} or vice-versa), they form a $k$-clique in terms of the connectivity structure of the visibility sub-graph for these $k$ guards in \G (i.e. each of the $k$ guards are visible to one another). %
Note that, to move from a guard in a particular {\it chamber} to one in a {\it corridor} or vice-versa, requires a traversal of at most 3 links in the visibility graph of \G\, namely, one link to move from the guard in a {\it chamber}'s clique to the guard in the chamber linked to a guard in the required {\it corridor}, the {\it chamber-corridor} link, and then one more link to the required guard in that {\it corridor}'s clique. %
Additionally, the sequence of {\it chambers} and {\it corridors} visited along the $\tilde{d}_v$ diameter realizing path is the same as the sequence of the corresponding nodes and edges of \T's diameter realizing path. Thus, we see that for this construction of $P$, $c \cdot d_{\T} \geq \tilde{d}_v \geq d_{\T}$ where $c$ is a small constant\footnote{ the later inequality follows trivially from the necessity of there being at least one guard in each chamber and corridor to ensure a minimal v-configuration of $P$.} giving us $d_{\T} \in \Theta(\tilde{d}_v)$. %

We see that given the tree to be explored has $n$ nodes, then, the obtained polygon $\mathcal{P}$ has at most $4(n-1)$ vertices. As the agents explore and traverse through the {\it corridors} of the polygon, it is equivalent to the agent traversing an edge in the tree exploration. %
Since we know, for $D \in o(\log n)$ (resp. $\tilde{d}_v \in o(\log n)$), from \cite{DISSER2018} (Theorem 2.5) that exploring a tree  of $n$ vertices using $\Theta(n)$ agents takes $\Omega(d_\T^2) = \Omega(D^2)$ time (or $\Omega(d_\T^2) = \Omega(\tilde{d}_v^2)$ time resp.), any algorithm for polygon exploration guard placement with $\Theta(n)$ agents on some polygon of $4(n-1)$ vertices requires $\Omega(D^2)$ (or $\Omega(\tilde{d}_v^2)$) rounds.
}

\section{Conclusion and Future Works}\label{conclusion}
\enlargethispage{2\baselineskip}
In this paper, we have presented centralized and distributed algorithms for computing a visibly connected guard placement. Crucially, our algorithms take time that is quadratic in a couple of different notions of diameters of $P$, i.e., $\tilde{d}_v$ and $D$.  We believe that $\tilde{d}_v \in O(D)$, thereby obviating the need for precise depth perception, but the proof has eluded us. It would be nice to establish this formally as this would lend to understanding the trade off between LiDAR and photogrammetry (see~\cite{FLR19,LorP} for example). We also remark that our algorithms have been explained in the synchronous setting for the purpose of clarity, but they can be easily extended to the asynchronous setting.

Additionally, it will be interesting to extend our works to more complicated structures like polygons with holes or polyhedra in dimensions greater than 2.

\paragraph{Acknowledgements} Barath Ashok, and John Augustine were supported in part by DST/SERB Extra Mural Grant (file number EMR/2016/00301) and DST/SERB MATRICS Grant (file number MTR/2018/001198). Suman Sourav was supported in part by the National Research Foundation, Prime Minister’s Office, Singapore under the Energy Programme and administrated by the Energy Market Authority (EP Award No. NRF2017EWT-EP003-047). Part of this work was done when Srikkanth Ramachandran and Suman Sourav visited IIT Madras. We thank Rajsekar Manokaran for pointing out LiDAR and Photogrammetry. 

\small
\bibliography{references}

\end{document}